\documentclass[journal,draftcls,onecolumn,12pt,twoside]{IEEEtran}
\usepackage{algorithm} %ctan.org\pkg\algorithms
\usepackage{algpseudocode}
\usepackage{graphicx}
\usepackage{amssymb}
\usepackage{amsfonts}
\usepackage{amsmath}
\usepackage{float}
\usepackage{url}
\usepackage{bm}
\usepackage[bottom]{footmisc}
 \usepackage[colorlinks=true]{hyperref}
\hypersetup{urlcolor=blue, citecolor=blue}

\IEEEoverridecommandlockouts

\newtheorem{Corollary}{Corollary
}

\newtheorem{Example}{Example}
\newtheorem{theorem}{Theorem}
\newtheorem{definition}{Definition}
\newtheorem{lemma}{Lemma}

\newenvironment{proof}{\textit{Proof\,:}} { $\blacksquare$}

\ifCLASSINFOpdf
\else
\fi
\begin{document}
\title{Construction of protograph-based LDPC codes \\
with chordless short cycles}
%-----------------------------------------------------------main affiliation
%\author{
%\IEEEauthorblockN{Hassan~Khodaiemehr}
%\IEEEauthorblockA{Department of Mathematics and Computer \\ Science,
%Amirkabir University of Technology\\
% E-mail: h.khodaiemehr@aut.ac.ir}
% \and
%\IEEEauthorblockN{Mohammad-Reza~Sadeghi}
%\IEEEauthorblockA{Department of Mathematics and Computer \\ Science,
%Amirkabir University of Technology\\
%Email: msadeghi@aut.ac.ir}
%\and
%\IEEEauthorblockN{Amin~Sakzad}
%\IEEEauthorblockA{Department of ECSE, Monash\\
%University, Victoria, Australia\\
%E-mail: amin.sakzad@monash.edu}
%}
%--------------------------------------------------------------
\author{Farzane Amirzade, Mohammad-Reza~Sadeghi and Daniel Panario\\

\thanks{F. Amirzade and  M.-R. Sadeghi are with the Department of Mathematics and Computer Science, Amirkabir University of Technology.  D. Panario is with the School of Mathematics and Statistics, Carleton University. 

(e-mail:  famirzade@gmail.com,  msadeghi@aut.ac.ir, daniel@math.carleton.ca).}
}

%\and
%\IEEEauthorblockN{Emanuele~Viterbo}
%\IEEEauthorblockA{Dept of ECSE, Monash University, Australia\\
%E-mail: emanuele.viterbo@monash.edu}}

%\author{Amin~Sakzad, and Mohammad-Reza~Sadeghi\thanks{A. Sakzad is with the ECSE Department at Monash University, Melbourne, Australia. M. R. Sadeghi is with the Faculty of Mathematics and Computer Science, Amirkabir University of Technology, Tehran, Iran.
%E-mails: amin.sakzad@monash.edu and msadeghi@aut.ac.ir.}}
\maketitle
\begin{abstract}
Controlling small size trapping sets and short cycles can result in 
LDPC codes with large minimum distance $d_{\min}$. We 
prove that short cycles with a chord are the root of several trapping 
sets and eliminating these cycles increases $d_{\min}$. We show that 
the lower bounds on $d_{\min}$ of an LDPC code with chordless short 
cycles, girths 6 (and 8), and column weights $\gamma$ (and 3), 
respectively, are $2\gamma$ (and 10), which is a significant 
improvement compared to the existing bounds $\gamma+1$ (and 6). 
Necessary and sufficient conditions for exponent matrices 
of protograph-based LDPC codes with chordless short cycles are proposed for any type of 
protographs, single-edge and multiple-edge, regular and irregular. The 
application of our method to girth-6 QC-LDPC codes shows that the removal of those cycles improves previous results 
in the literature. 

\end{abstract}

% \fi
\IEEEpeerreviewmaketitle
%\vspace{-0.4cm}
\section{Introduction}
\IEEEPARstart{A} protograph   (Ptg) is a small size bipartite graph whose adjacency matrix is taken as a base matrix of  protograph-based LDPC codes. These codes  which are associated with a base matrix and an exponent matrix 
have attracted attention. It is known that the length of the shortest 
cycles of the Tanner graph (TG), $\emph{girth}$, influences code performance.  Graphical structures including detrimental 
$\emph{trapping sets}$ (TS) are known to be key factors of error floor 
behavior of LDPC codes. An $(a,b)$ TS is a set of $a$ variable-nodes, 
v-nodes, in the TG which induce a subgraph of the TG with $b$ 
check-nodes, c-nodes, of odd degrees and an arbitrary number of even 
degree c-nodes. An $(a,b)$ TS is $\emph {elementary}$ (ETS) if all c-nodes 
are of degree 1 or 2.

Removing TSs up to a certain size which are subgraphs of $(a,0)$ TSs, 
yields a code with improved minimum distance $d_{\min}$. 
In \cite{Battag2018}, it is shown that controlling specific cycles  
contributes to the removal of harmful trapping sets. In \cite{AMC} is 
given,   using an edge-coloring technique, a sufficient condition for an exponent matrix to give a 
$(3,n)$-regular algebraic-based QC-LDPC code with girth $g=6$,   column weight $\gamma=3$, row weight $n$ and free 
of $(a,b)$ ETSs, $4\leq a\leq 5$ and $b\leq2$.   Most methods in the literature to remove harmful TSs are applied to fully-connected regular Ptgs; there are not many results regarding irregular and multiple-edge Ptgs. Moreover, they focus on removing a specific TS whereas,   numerical results in \cite{Karimi2019} show 
that for $g=6$ and $\gamma=3$ the removal of ETSs of large size  yields a minimum lifting degree which is larger than the minimum 
lifting degree of a girth-8 code with the same degree distribution.

  In this paper, we show that short cycles with a chord, cycle-wc, are the root of several trapping 
sets and eliminating these cycles increases $d_{\min}$. We provide a new method to control short cycles-wc which is applicable to any exponent matrix, 
does not cause a lifting degree larger than the minimum lifting degree of an LDPC code with the same degree distribution but a higher 
girth, and improves the lower bound on $d_{\min}$. We prove that our approach can be applied to single-edge (SE) and multiple-edge (ME) Ptgs, and it is not limited to regular codes. We provide numerical and analytical results to show the applicability of our method to irregular codes such as Raptor-Like codes \cite{Divsalar} and LDPC codes from Sidon sequences \cite{Sidon}. Our main contributions are as follows:
(1) For $g=6$ and   minimum column weight $\gamma$ we show that avoiding 8-cycles-wc results in an LDPC code in which the lower bounds on the size of 
smallest $(a,b)$ ETSs, $b<a$, are equal to the minimum sizes of 
ETSs in a girth-8 LDPC code, and the lower bound on $d_{\min}$ is 
increased from $\gamma+1$ to $2\gamma$. (2) For $g=8, \gamma=3$ we prove that avoiding 12-cycles-wc yields an LDPC code 
with $d_{\min}\geq10$ and the lower bound on the size of the 
smallest $(a,b)$ ETSs, $b<a$, is equal to the lower bound on the 
size of ETSs of a girth-10 LDPC code. (3) We obtain a necessary and sufficient condition to construct 
girth-6 LDPC codes free of 8-cycles-wc.

We give the structure of the paper. Section \ref{II} presents some 
basic definitions. Section \ref{III} investigates cycle-wc of short 
lengths. Section \ref{IV} presents a necessary and sufficient 
condition to remove 8-cycles-wc for   both SE and ME Ptg-based LDPC codes along with their numerical results.   In Section \ref{V}, we 
summarize our results. 

%\vspace{-0.3cm} 

\section{Preliminaries}\label{II}
  A QC-LDPC code with a lifting degree $N$ can be associated to an exponent matrix $B=[\vec{B}_{ij}]$ and a base matrix $W=[W_{ij}]$ of the same size $c\times d$, where for $0\leq i\leq c-1$ and $0\leq j\leq d-1$, $|\vec{B}_{ij}|=W_{ij}$,  $\vec{B}_{ij}=(b^{1}_{ij},b^{2}_{ij},\ldots, b^{l}_{ij})$, $b^{r}_{ij}\in \lbrace 0,1,\ldots,N-1\rbrace$ and $\;b^{r}_{ij}\neq b^{r^\prime}_{ij}$ for $\;1\leq r<r^\prime\leq l,\;l\in {\mathbb N}$. If $W_{ij}=0$, then  $B_{ij}=(\infty)$. If Ptg is SE, then all elements of $W$ are 0,1. If Ptg is ME, then $W$ contains entries bigger than 1.  Substituting vectors of $B$  by $ H_{ij}=I^{b^{1}_{ij}}+I^{b^{2}_{ij}}+\cdots +I^{b^{l}_{ij}},$ where $I^{b^{r}_{ij}}$ is a circulant permutation matrix (CPM) of dimension $N\times N$, and  $\infty$ elements by a zero matrix of size $N\times N$ yields a parity-check matrix whose null space  gives a QC-LDPC code. The top row of $I^{b^{r}_{ij}}$ contains 1 in the $b^{r}_{ij}$-th position and other entries of that row are zero. The $s$-th row of CPM is formed by $s$ right cyclic shifts of the first row. A necessary and sufficient condition for the existence of $2k$-cycles in  the TG of QC-LDPC codes provided in \cite{2004} is\begingroup\fontsize{8.5pt}{8.5pt}\begin{align}\label{Equation}
	\sum_{i=0}^{k-1} \left( b^{r^{}_i}_{m_{i}n_{i}} - b^{r^{\prime}_i}_{m_{i}n_{i+1}} \right)\equiv0\pmod {N} , 
\end{align}\endgroup
where $b^{r^{}_i}_{m_{i}n_{i}}$ is the $r^{}_i$-th entry of $\vec{B}_{m_{i}n_{i}}$, $\vec{B}_{m_{i}n_{i}}$ is the $(m_i n_i)$-th entry of $B$, $n_{k}=n_{0}$, $r^{}_{i}\neq r^{\prime}_{i}$ if $n_{i} = n_{i+1}$, and $ r^{\prime}_{i}\neq r^{}_{i+1}$ if $m_{i} =m_{i+1}$.
%In protograph-based QC-LDPC codes, an exhaustive search is required to construct  an $m\times n$ exponent matrix $B$. If the lifting degree is $N$, then the size of the search space for  $B$ is $N^{mn}$. One of the merits of the compact method is that the size of the search space is considerably reduced to $(N-2)^{m+n-4}$ ; see \cite{Tadayon}. This contributes to a  complexity reduction when obtaining exponent matrices with large $n$ and $N$. 
\begin{definition}
 For a bipartite graph ${\it G}$ corresponding to an ETS, a $\emph{variable node}$ (VN) $\emph{graph}$ is constructed by removing all degree-1 c-nodes, defining  v-nodes of ${\it G}$ as its vertices and degree-2  c-nodes  connecting the v-nodes in ${\it G}$   becoming edges. 
 \end{definition}
 \begin{definition}
  A cycle is a  $\emph{chordless cycle}$,  denoted by cl-cycle, if outside the cycle there  is no connection between two v-nodes of it, or between a 
v-node and a c-node of that cycle.  Otherwise, it is defined as a 
$\emph{cycle with a chord}$ which is denoted by cycle-wc.
\end{definition} 

 %\vspace{-0.4cm}	
\section{Chordless Short Cycles}\label{III}
In this section, we consider cycles in terms of their chords. We prove 
that in a girth-6 LDPC code all 6-cycles, and in a girth-8 LDPC code 
all 8-cycles and 10-cycles, are free of a chord. We show that without increasing
girth we can improve significant features of a code such as 
$d_{\min}$ and TSs which influence the performance of a code in the 
error floor region.  
\begin{theorem}\label{chordlesscycles}
  In a girth-6 TG, 6-cycles are chordless; an 8-cycle has a chord if only if two v-nodes of the cycle are connected to only a common c-node outside the cycle. In a girth-8 TG, 8-cycles and 10-cycles are chordless; 12-cycles 
have a chord if only if two v-nodes of the cycle are connected to % only 
a common c-node outside the cycle. 
\end{theorem}
\begin{proof}
  We consider only the 8-cycle in a girth-6 TG case; the other proofs are similar. Suppose $v_1c_1v_2c_2v_3c_3v_4c_4v_1$ is an 8-cycle, $C$.  The existence of one of the edges $v_1c_3,v_1c_2$, $v_2c_3,v_2c_4$, $v_3c_1,v_3c_4$, $v_4c_1$ or $v_4c_2$ outside $C$ causes a 4-cycle. Hence, connecting a c-node and v-node outside $C$ is impossible and  chords in an 8-cycle occur when two v-nodes are connected to a common c-node outside the cycle. Now, if those v-nodes have also a common connection in the cycle, then a 4-cycle is obtained. Thus, the existence of a c-node between $v_1,v_3$ or between $v_2,v_4$ is necessary and sufficient to have an 8-cycle-wc. \hfill
\end{proof}

\begin{theorem}\label{chord8cycle}
Let a 4-cycle free TG with column weight $\gamma$ be given. Then, in an $(a,b)$ ETS with cl-8-cycles
\begin{itemize}
\item   the number of edges of a VN graph is $|E|\leq\frac{a^3}{4a-3}$;
\item the parameters $a,b,\gamma$ satisfy the inequality 
  $b\geq a\gamma-\frac{2a^3}{4a-3}$, and if $b<a$, then   $a\geq2\gamma-2$.
\end{itemize} 
\end{theorem}
\begin{proof}
  Let $K_n$ be the complete graph on $n$ vertices. A $K_4$-free graph on $n$ vertices has at most $\frac{n^2}{3}$ edges \cite{Extermal}. The VN graph of a cl-8-cycle $(a,b)$ ETS is $K_4$-free, so its maximum number of edges is $\frac{a^2}{3}$. Moreover, a triangle-free graph on $n$ vertices has at most $\frac{n^2}{4}$ edges \cite{Extermal}. If we include the number of edges of all $K_3$-free VN graphs of  $(a,b)$ ETSs  in a set $R$ and put the number of edges of all VN graphs of ETSs free of 8-cycles-wc in a set $Y$, then $R\subset Y$. Thus, the maximum integer in $Y$ is bigger than or equal to the maximum integer in $R$, which is $\frac{a^2}{4}$. Hence, for a VN graph of an $(a,b)$ ETS free of 8-cycles-wc and with maximum edges $|E|$, we have $|E|\geq\frac{a^2}{4}$. Thus, the VN graph of a cl-8-cycle $(a,b)$ ETS with the maximum number of edges satisfies $\frac{a^2}{4}\leq|E|\leq\frac{a^2}{3}$ in which the number of triangles is at least $\frac{a}{9}(4|E|-a^2)$; see \cite{Extermal}. Since the VN graph is free of the graphs in Fig. \ref{FigCW3} (a) and (b), no two triangles in the VN graph have an edge in common. Thus, the maximum number of triangles is $\frac{|E|}{3}$. Hence,  $\frac{|E|}{3}\geq\frac{a}{9}(4|E|-a^2)$ that yields $|E|\leq\frac{a^3}{4a-3}$.  

  In an LDPC code with column weight $\gamma$, an $(a,b)$ ETS whose 
VN graph contains $|E|$ edges satisfies the equality $b=a\gamma-2|E|$ and since the TG is 4-cycle free, $a\geq\gamma+1$ ;
see \cite{TCOM}. Since all 8-cycles in the ETS are chordless, 
$|E|\leq\frac{a^3}{4a-3}$. Hence, $b \geq a\gamma-\frac{2a^3}{4a-3}$. 
Now, if $b<a$, then we have
% from the inequalities $b\geq a\gamma-\frac{a^2}{2}$ 
% and $\frac{b}{a}<1$ we obtain 
$1>\frac{b}{a}>\gamma-\frac{2a^2}{4a-3}$ implying  $4a-4a\gamma+2a^2>3-3\gamma$.  A case analysis shows that the last inequality does not hold for $\gamma\geq3$  and 
$\gamma+1\leq a\leq 2\gamma-3$. For $\gamma=1,2$, we have $a\geq 1$ and $a\geq 3$, respectively. Hence, for each $\gamma$ we have 
 $a\geq2\gamma-2$. \hfill
\end{proof}
%\vspace{-0.5cm}
\begin{center}
\begin{figure}
\centering
\includegraphics[scale=.3]{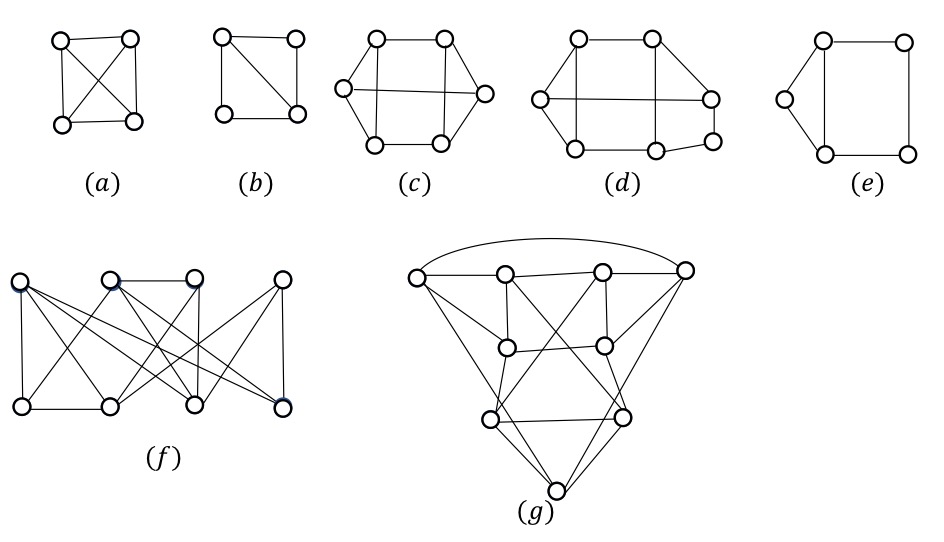}
\caption{ Figures (a), (b) are the VN graphs of $(4,0)$ and $(4,2)$ ETSs (also VN graphs of an 8-cycle-wc). Figures (c), (d) and (e) are the VN graphs of $(6,0),\ (7,1), (5,3)$ ETSs with $\gamma=3$. Figures (f) and (g) are the VN graphs of $(8,4)$  and $(9,0)$ ETSs with $g=6,\gamma=4$.} \label{FigCW3}
\end{figure}
\end{center} 
%\vspace{-0.3cm}
\begin{Example}\label{ExCW3}
Let $\gamma=3$. For different $b$s, we find the smallest size of an 
ETS containing at least one 6-cycle and free of an 8-cycle-wc. Suppose $b=0$, 
then the smallest $a$ satisfying   $b\geq a\gamma-\frac{2a^3}{4a-3}$ is 
$a=6$ whose VN graph with maximum number of edges is in Fig. 
\ref{FigCW3} (c). Although for $b\geq1$ the smallest $a$ satisfying 
  $1\geq3a-\frac{2a^3}{4a-3}$ is $a=5$, the minimum $a$ obtained for $b=1,2$ is $7,6$, respectively. The VN graph with maximum number of 
edges of $(7,1)$ ETS with $g=6$ and free of an 8-cycle-wc is in Fig. \ref{FigCW3} (d). Removing an edge from Fig. \ref{FigCW3} 
(c) results in a VN graph of a $(6,2)$ ETS. Fig. \ref{FigCW3} (e) is the VN graph of a 
$(5,3)$ ETS. 
\end{Example}

According to Theorem \ref{chordlesscycles}, in a girth-8 TG all 8-cycles 
are chordless. Hence, if for specific $\lambda$s and $b$s the 
smallest $a$ for a 6-cycle free ETS is less than the one of an ETS 
containing 6-cycles but free of an 8-cycle-wc, then we consider the 
former as the smallest $a$ in a TG with cl-8-cycles.
\begin{Example}\label{smallestsize}
Let $\gamma=4$. If $b=0$, then in a girth-8 LDPC code we have $a\geq8$; 
see \cite{TCOM}. Now, we prove that if an ETS with cl-8-cycles contains 
a 6-cycle, then the smallest size is more than 8. Since $b=0$, all 
v-nodes of the VN graph have degree 4. Suppose $v_0$ is a v-node 
connected to four c-nodes $c_1,c_2,c_3,c_4$. Since all c-nodes have 
degree 2, each c-node is connected to other v-node. Let $v_i$ be 
connected to $c_i$. Without loss of generality, suppose $v_0,v_1,v_2$ 
are the v-nodes of a 6-cycle. If a c-node connects $v_1$ or $v_2$ to 
one of the v-nodes $v_3,v_4$, then an 8-cycle-wc is obtained which 
is impossible. Thus, there must be v-nodes $v_5$ and $v_6$ for two 
different c-nodes connected to $v_1$. If there is a c-node between 
$v_2$ and $v_5$ or between $v_2$ and $v_6$, then we have an 8-cycle-wc 
with v-nodes $v_0,v_1,v_5,v_2$ or $v_0,v_1,v_6,v_2$, respectively. 
Therefore, there is no c-node between $v_2$ and any of the v-nodes 
$v_3,v_4,v_5,v_6$. Hence, there must be two v-nodes $v_7,v_8$ for 
two different c-nodes connected to $v_2$. This process guarantees the 
non-existence of an $(8,0)$ ETS containing 6-cycles and cl-8-cycles.
\end{Example}

In general, in an LDPC code with $g=6, \gamma=4$  all $(a,b)$ ETSs, $5\leq a\leq8,b\leq2$, and $(6,4),(7,4)$ ETSs  which contain 6-cycles have at least one 8-cycle-wc.  Fig. \ref{FigCW3} (f) and (g) show that there are $(8,4)$ and $(9,0)$ ETSs including 6-cycles and  cl-8-cycles. Since an LDPC code with $g=8,\gamma=4$ is free of all $(a,b)$ ETSs, $5\leq a<8,b\leq2$, and $(6,4)$ ETSs but contains $(8,0), (8,2), (7,4)$ ETSs, the minimum size of an ETS with cl-8-cycles for $b=0,2,4$ is $a=8,8,7$, respectively. Similar to Examples \ref{ExCW3}, \ref{smallestsize} we use Theorem \ref{chord8cycle} to present the lower bound on the size of an ETS  in Table \ref{Tabchord8-cycle} for an  LDPC code with $3\leq\gamma\leq6$.   The results of Theorem \ref{chordlesscycles} can be also extended to irregular LDPC codes.

\begin{theorem}
  Let an irregular LDPC code with column weights $\{d_1,\ldots,d_{\ell}\}$ and free of 8-cycles-wc be given where  the minimum v-node degree is $d_1=\gamma$. The  lower bound on the smallest size of an $(a,b)$ ETS, $b<a$, is $2\gamma-2$.
\end{theorem}
\begin{proof}
  We consider an $(a,b)$ ETS, $b<a$, free of 8-cycles-wc and remove $d_i-\gamma$ degree-1 c-nodes from $i$-th v-node with $d_i>\gamma$. The resulting structure is an $(a,b')$ ETS with $b'<b<a$ which belongs to an LDPC code with column weight $\gamma$ and free of 8-cycles-wc. According to Theorem \ref{chord8cycle}, for the $(a,b')$ ETS with $b'<a$ we have $a\geq2\gamma-2$.\hfill
\end{proof}
%\vspace{-0.15cm}

\begin{table}[ht]
\setlength{\tabcolsep}{.7 pt}
\centering
\caption{ The lower bounds obtained on the size of  $(a,b)$ ETSs of variable-regular LDPC codes with $3\leq\gamma\leq6$, $g=6$  and cl-8-cycles.}
%\vspace{-.7em}
\begingroup\fontsize{8.5 pt}{8.5 pt}
\begin{tabular}{|l|c|c|c|c|c|c|c|c|c|c|c|c|c|c|c|c|c|c|c|c|c|c|} \hline
$\gamma=$&3&3&3&3&4&4&4&4&4&5&5&5&5&5&5&6&6&6&6&6&6&6\\ \hline    
$b=$ & 0&1&2&3&0&1&2&3&4&0&1&2&3&4&5&0&1&2&3&4&5&6\\ \hline
$a\geq$ & 
% 6&7&6&5&8&None&8&None&7&10&11&10&11&10&9&12&None&12&None&12&None&11  \\
6&7&6&5&8&$-$&8&$-$&7&10&11&10&11&10&9&12&$-$&12&$-$&12&$-$&11  \\
\hline
\end{tabular}
\endgroup
\label{Tabchord8-cycle}
 %\vspace{-.7em}
\end{table}
\begin{theorem}\label{min1}
  For an LDPC code with column weight $\gamma\geq3$, free of 4-cycles 
and 8-cycles-wc we have $d_{\min}\geq2\gamma$.
\end{theorem} 
\begin{proof}
It is known that a code $\mathcal{C}$ has minimum 
distance $d_{\min}$ if and only if the TG contains no $(a,0)$ TS 
for $a<d_{\min}$ and there exists at least one $(d_{\min},0)$ TS. 
Since the TG is free of an 8-cycle-wc,   according to Theorem 
\ref{chord8cycle}, 
$a\geq2\gamma-2,b\geq a\gamma-\frac{2a^3}{4a-3}$ which does not hold for $b=0, \gamma\geq3$ if $a=2\gamma-2$ or $a=2\gamma-1$. Thus, that has a hanging $2\gamma-1$. 
Therefore, there is no $(a,0)$ 
ETS for $a<2\gamma$. Now, we prove the non-existence of 
$(a,0)$ non-elementary TSs for $a<2\gamma$. These TSs contain $2\ell$-degree 
c-nodes, $\ell>1$. We consider a structure starting from a c-node. 
Suppose TS contains a 4-degree c-node $c$ in the starting level. 
Then, $c$ is connected to 4 v-nodes $v_1,v_2,v_3,v_4$ in the 
second level. Since every v-node has degree $\gamma$, each v-node 
connected to $c$ is connected to $\gamma-1$ other c-nodes. If two 
of these v-nodes such as $v_1,v_2$ are connected to a common c-node 
$c'$, then a 4-cycle $cv_1c'v_2c$ appears which is impossible. 
Therefore, there are $4(\gamma-1)$ distinct c-nodes in the third 
level. Similarly, if two c-nodes $c_1,c_2$ in the third level, 
which are connected to a common v-node $v$ of the second level, 
are connected to another v-node $v'$, then a 4-cycle $vc_1v'c_2v$ 
appears which is impossible. Hence, in the fourth level there are 
at least $\gamma-1$ distinct v-nodes. If a v-node $v$ in the fourth 
level is connected to two c-nodes $c',c''$ in the third level such 
that $c',c''$ are connected to two different v-nodes like $v_1,v_2$ 
of the second level then TS has a 6-cycle $vc'v_1cv_2c''v$. Thus, 
if the TG is free of 6-cycles, then we have $4(\gamma-1)$ v-nodes 
in the fourth level and the size of TS is more than or equal to 
$4(\gamma-1)+4=4\gamma$. Let the TS have a 6-cycle. As mentioned 
above, in the fourth level there are $\gamma-1$ distinct v-nodes 
for $\gamma-1$ c-nodes connected to $v_1$. If two of these v-nodes 
are connected to c-nodes in the third level connected to $v_2$, 
then we have an 8-cycle-wc, see Fig. \ref{MINDIS} (a). Considering 
this restriction, in Fig. \ref{MINDIS} (b), (c) we show that the 
smallest sizes of a TS for $\gamma=3,4$ are 8 and 10, respectively. 
For $\gamma\geq5$ we can use the TS in Fig. \ref{MINDIS} (c) and 
add $(\gamma-4)$ c-nodes to each v-node of the second level. 
Considering the mentioned restriction, v-nodes in the fourth level 
cannot be connected to a new c-node in the third level. Each of 
these new c-nodes must be connected to a new v-node in the fourth 
level. Up to now, the smallest size of a TS is at least 
% more than or equal to 
$10+4(\gamma-4)=4\gamma-6>2\gamma$ for  $\gamma\geq5$. 
Thus, we need at least $2\gamma$ v-nodes to construct a TS 
containing a 4-degree c-node. Hence, $d_{\min}$ is lower bounded 
by $2\gamma$. \hfill
\end{proof}

%\vspace{-0.5cm}
\begin{center}
\begin{figure}
\centering
\includegraphics[scale=.3]{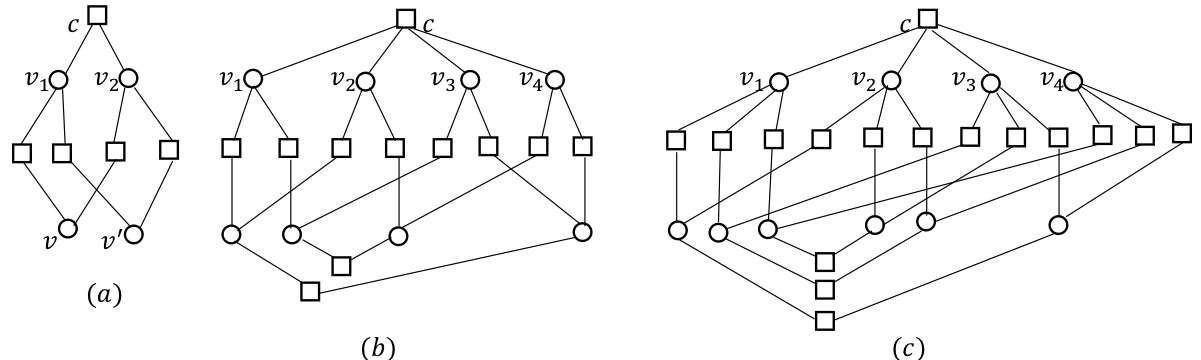}
\caption{ Figure (a) shows an 8-cycle-wc, (b) is an (8,0) non-elementary TS with $\gamma=3$, (c) is a (10,0) non-elementary TS with $\gamma=4$.} \label{MINDIS}
\end{figure}
\end{center} 
%\vspace{-0.3cm}

  It should be noticed that avoiding chords is not the only way to improve $d_{\min}$. In fact, the graphical results in \cite{Zhao2016} showed that array-based codes with $\gamma=3$ contain 8-cycles-wc and indicated that it is possible for LDPC codes with $\gamma=3$ and 8-cycles-wc  to have $d_{\min}=6$.
\begin{lemma}\label{minsizeg8}
The minimum sizes of an $(a,b)$ ETS, $b<a$, in an LDPC code 
with $g=8,\gamma=3$ and cl-cycles of lengths up to 12 for $b=0,1,2,3$ 
are $a=10,9,8,5$, respectively.
\end{lemma}
\begin{proof}
The smallest sizes of an $(a,b)$ ETS with $b<a$ in an LDPC 
code with $g=8,\gamma=3$ for $b=0,1,2,3$ are $a=6,5,4,3$; see \cite{TCOM}. 
According to Theorem \ref{chordlesscycles}, if $g=8$, then the shortest cycles-wc are of length 12. Avoiding 12-cycles-wc causes the removal 
of all $(6,0), (6,2), (7,1), (8,0)$ ETSs. Consequently, if $b<a$ the 
smallest ETSs for different $b$s are $(10,0), (9,1), (8,2), (5,3)$ 
ETSs.\hfill
\end{proof}

\begin{theorem}\label{mindisg8}
The $d_{\min}$ of an LDPC code with $g=8,\gamma=3$  and  cl-cycles of lengths up to 12 is lower bounded by 10. 
\end{theorem}
\begin{proof}
According to Lemma \ref{minsizeg8}, the minimum size of $(a,0)$ ETSs is 10. By the proof of Theorem \ref{min1}, if $g=8$, then for $(a,0)$ non-elementary TSs we have  $a\geq 4\gamma$ and for  $\gamma=3$ we have $a\geq12$ . Thus, $d_{\min}\geq10$. \hfill
\end{proof}
%\vspace{-0.15cm}

\section{Necessary and sufficient condition to remove 8-cycles with a chord}\label{IV}
% In this section, 
%We present a necessary and sufficient condition to construct a  Ptg-based LDPC code with girth 6 and cl-8-cycles. 

The VN graph of an 8-cycle-wc consists of two triangles which have one edge in common, Fig. \ref{FigCW3} (b). Each triangle in the VN graph corresponds to a 6-cycle in the TG. Thus,  a necessary and sufficient condition to remove an 8-cycle-wc is to avoid   the occurrence of two 6-cycles with one common c-node. The constraints to avoid these 6-cycles are given in Subsections \ref{A} and \ref{B}.  We  propose this condition for the exponent matrix to provide an explicit technique to avoid these 8-cycles-wc.  
\vspace{-0.25cm}
%\vspace{-.7em}
\subsection{Submatrices of $B$ of single-edge protographs}\label{A}
When checking 6-cycles by Equation (\ref{Equation}), if two equations with zero in their  right-hand side (equivalent to two 6-cycles) have one equal term like $b^{r^{}_i}_{m_{i}n_{i}} - b^{r^{\prime}_i}_{m_{i}n_{i+1}}$, then these 6-cycles cause an 8-cycle-wc in the TG.   To avoid the occurrence of two 6-cycles with one common c-node regarding SE-Ptgs (where $|W_{ij}|\leq1$ and $B_{ij}=(b_{ij})$ for $|W_{ij}|=1$) all $3\times3$ and $3\times4$ submatrices of $B$ are checked.

First, we present all six equations regarding Equation (\ref{Equation}) to check 6-cycles in a $3\times3$ submatrix of $B$: 
\begin{enumerate}
	\item $(b_{i_1j_1}-b_{i_1j_2})+(b_{i_2j_2}-b_{i_2j_3})+(b_{i_3j_3}-b_{i_3j_1})=e_1$,
	\item $(b_{i_1j_1}-b_{i_1j_3})+(b_{i_2j_3}-b_{i_2j_2})+(b_{i_3j_2}-b_{i_3j_1})=e_2$,
	\item $(b_{i_1j_2}-b_{i_1j_3})+(b_{i_2j_3}-b_{i_2j_1})+(b_{i_3j_1}-b_{i_3j_2})=e_3$,
	\item $(b_{i_1j_3}-b_{i_1j_2})+(b_{i_2j_2}-b_{i_2j_1})+(b_{i_3j_1}-b_{i_3j_3})=e_4$,
	\item $(b_{i_1j_3}-b_{i_1j_1})+(b_{i_2j_1}-b_{i_2j_2})+(b_{i_3j_2}-b_{i_3j_3})=e_5$,
	\item $(b_{i_1j_2}-b_{i_1j_1})+(b_{i_2j_1}-b_{i_2j_3})+(b_{i_3j_3}-b_{i_3j_2})=e_6$.
\end{enumerate} 

Then, for each $3\times3$ submatrix of $B$, we consider a 6-entry 
vector $e=[e_1,e_2,\dots,e_6]$. 
\begin{theorem}\label{3by3}
A necessary and sufficient condition to avoid an 8-cycle-wc in each 
$3\times3$ submatrix of $B$ is that its equivalent 6-entry vector 
$e$ does not satisfy the following equalities:
$
\begin{array}{lll}
 (1)\ e_1=e_6=0,& (2)\ e_1=e_2=0,& (3)\ e_1=e_4=0,\\
 (4)\ e_2=e_5=0,& (5)\ e_2=e_3=0,& (6)\ e_3=e_4=0,\\
 (7)\ e_3=e_6=0,& (8)\ e_4=e_5=0,& (9)\ e_5=e_6=0.\\
 \end{array}
$
\end{theorem} 

\begin{proof}
To avoid an 8-cycle-wc we only check those 6-entry vectors which contain at least two zeros. Now, if $e_1=e_2=0$, then, since the first two equations contain a term $(b_{i_2j_2}-b_{i_2j_3})$ or $-(b_{i_2j_2}-b_{i_2j_3})$, these two 6-cycles belong to an 8-cycle-wc. Similarly, if one of equalities $(1)$ to $(9)$ occurs, then the TG contains an 8-cycle-wc.\hfill
\end{proof}
\begin{definition}\label{Compact}
 Let $\vec{0}$ be an all-zero column vector of size $\gamma$, $\vec{B_1}=[0,1,b_{21},\ldots,b_{(\gamma-1)1}]$ be a $\gamma$-entry column vector, where $b_{i1}\in\{2,\ldots,N-1\}$, $\gamma_j\otimes \vec{B_1}$ be the $j$-th column vector, where $\gamma_j\in\{2,\ldots,N-1\},\ \gamma_j<\gamma_{j+1}$ and $\otimes$ denote the multiplication modulo $N$. These $n$ column vectors form the following $\gamma\times n$ exponent matrix of a compact QC-LDPC code  
$
B=[\vec{0}\vert \vec{B_1}\vert \gamma_2\otimes\vec{B_1}\vert\cdots\vert\gamma_{n-1}\otimes\vec{B_1}]	$ with lifting degree $N$.
\end{definition}

  Applying Theorem \ref{3by3} to an exponent matrix of a compact QC-LDPC code yields the following result.
\begin{Corollary}\label{Compact3by3}
  Let an exponent matrix $B$ of a compact QC-LDPC with a vector of coefficients $[0,1,\gamma_2,\ldots,\gamma_{n-1}]$ and the vector $\vec{B}_1$ be given. If for each three elements $p,q,r\in\vec{B_1}$ and $0\leq i,j\leq n-1$ three phrases $(p+q-2r)(\gamma_i-\gamma_j),\ (p+r-2q)(\gamma_i-\gamma_j)$ and $(q+r-2p)(\gamma_i-\gamma_j)$ are nonzero modulo $N$, then no $3\times3$ submatrix of $B$ causes an 8-cycle-wc.
\end{Corollary}
%\vspace{-0.15cm}

In order to check 6-cycles in a $3\times4$ submatrix of $B$, for every zero element of a 6-entry vector we allocate a 
$\emph{column-index vector}$ containing three pairs of column indices 
of $B$ which appear in Equation (\ref{Equation}). For a $c\times d$ 
exponent matrix, the column-index vector has $c$ elements. Since for 
each vector there are three pairs of column indices, we take the empty 
set for the positions corresponding to the row indices of $B$ which 
are not among the three rows appearing in Equation (\ref{Equation}). 
For example, if Equation (\ref{Equation}) is checked for the first 
three rows of a $4\times d$ exponent matrix and three columns with 
indices $j_1,j_2,j_3$, then we assign $e_1=0$ to $[(j_1,j_2),(j_2,j_3),
(j_1,j_3),\emptyset]$.  
\begin{theorem}\label{3by4}
A necessary and sufficient condition to remove an 8-cycle-wc from 
each $3\times4$ submatrix of $B$ is that no two 6-entry vectors of 
it yield two column-index vectors containing a common pair in the 
same position. 
\end{theorem} 
\begin{proof}
A $3\times4$ submatrix of $B$ contains four $3\times3$ submatrices. Among  the $3\times4$ submatrices we consider those containing at least two $3\times3$ submatrices whose 6-entry vectors have a zero element. Suppose $U$ is a $3\times4$ submatrix of $B$, and $V$ and $W$ are two $3\times3$ submatrices of $U$ with the above property.  If the column-index vector assigned to one zero from the 6-entry vector of $V$ and the column-index vector allocated to one zero from the 6-entry vector of $W$ have a common pair in the same position, their corresponding 6-cycles belong to an 8-cycle-wc.\hfill
\end{proof}
%\vspace{-.7em}

Avoiding 8-cycles-wc from occurrence in the TG of a girth-6 LDPC code causes the removal of small size ETSs as well as increasing $d_{\min}$. For example, the elimination of 8-cycles-wc results in the removal of $(a,b)$ ETSs with $a\leq5$, $b\leq2$ and $d_{\min}\geq6$ for $\gamma=3$ and the removal of $(a,b)$ ETSs with $a\leq8$, $b\leq2$, $(6,4)$, $(7,4)$ ETSs and $d_{\min}\geq8$ for $\gamma=4$. Thus, a necessary and sufficient condition to construct a $(\gamma,n)$-regular QC-LDPC code with $g=6$ and $d_{\min}\geq2\gamma$  whose TG is free of small size $(a,b)$ ETSs, $b<a$, is to apply Theorems \ref{3by3} and \ref{3by4} to the exponent matrix. 

%More constraints are required to remove ETSs of larger size such as $(6,0)$ and $(6,2)$ ETSs from the TG of a girth-6 QC-LDPC code with  $\gamma=3$. Numerical results in \cite{Karimi2019} show that avoiding larger size ETSs causes an increase in the lifting degree, $N$, such that the obtained minimum $N$ is larger than the minimum one of a girth-8 QC-LDPC code with the same degree distribution. Therefore, if a QC-LDPC code with minimum $N$ and free of those large size ETSs is desired one may prefer a girth-8 code. 

Applying the necessary and sufficient conditions proposed in Theorems \ref{3by3} and \ref{3by4} to a $3\times5$ exponent matrix gives a $(3,5)$-regular QC-LDPC code with  $N=10,g=6$, and free of $(a,b)$ ETSs, $a\leq5$, $b\leq2$. In this case, if we consider the compact method in Definition \ref{Compact}, we obtain $N=11$. Since the compact method has less complexity and accelerates the process of finding an exponent matrix within a certain size, we prefer to use this method. In Tables \ref{TabC3}, \ref{TabC4}, we provide an exponent matrix  $B$ for a compact QC-LDPC code with $g=6$, different $\gamma$s and $n$s. We take the lifting degree achieved for each $n$ as an upper bound on $N$ of a  $(\gamma,n)$-regular QC-LDPC code with $g=6$ and free of 8-cycles-wc. In each case we only present the second row of  $B$ as well as $\vec{B_1}$.   Using the compact technique, Corollary \ref{Compact3by3}, Theorem \ref{3by4}, it takes the search algorithm less than 10 seconds to find most of $B$s.

As can be seen in Table \ref{TabC3}, for $\gamma=3$ and $n=5,7,8,9$, removing 8-cycles-wc results in girth-6 codes  whose $d_{\min}$s are equal to the maximum $d_{\min}$s reported in   \cite{Voltage} for girth-8 QC-LDPC codes with the same degree distribution. Moreover,  for $\gamma=4,n=6,7,$ $d_{\min}$ of girth-6 codes free of small size ETSs  are larger than those with $g=8$.     Another merit of removing 8-cycles-wc, and consequently the elimination of small size  ETSs, is shown in the performance curves of two $(3,6)$-regular QC-LDPC codes $C_1$, $C_2$ with $d_{\min}=8,N=13$ and two $(3,11)$-regular QC-LDPC codes $C_3$, $C_4$ with $d_{\min}=8,N=31$, such that, $C_1,C_3$ are free of the small ETSs and $C_2,C_4$ contain those ETSs.  As we expect $C_1$, $C_3$ outperform their counterparts $C_2$, $C_4$, respectively; see Fig. \ref{Fig2}. The codes simulated are so short that they do not show an error floor. Instead, because of the removal of 8-cycles-wc, we see a small improvement in waterfall. Fig. \ref{Fig2} also shows the performance curve of a $(4,12)$-regular QC-LDPC code, $C_5$,  free of 8-cycles-wc and $N=52$. Performances of these codes were decoded using the sum-product algorithm with 50 iterations. 
%\vspace{-0.5cm}
%\vspace{-0.5cm} 
%\vspace{-0.5cm}
\begin{table}[t]
 	\setlength{\tabcolsep}{.7 pt}
 	\centering
 	\caption{  An exponent matrix $B$  of a $(\gamma,n)$-regular compact QC-LDPC code free of 8-cycles-wc, with $g=6$ and minimum  $N$ whose $d_{\min}$ is compared with the  highest $d_{\min}$ of the code with $g=8$ in \cite{Voltage}.}

%\vspace{-0.25cm}

\begingroup\fontsize{8.5pt}{11pt}
\begin{tabular}{|l|l|l|c|c|c|l|}
 		\hline

 			$\gamma,n$  &  Second row of $B$ &  $\vec{B_1}$ & $N$ & $d_{\min}$&  Runtime  & $N,d_{\min},g=8$  \\
 		\hline

 		3,5  & $[0, 1, 2, 4, 7]$ & $[0,1,3]$& 11 & 10 & 0.03 & 13, 10 \\
 		\hline

 		3,6 	& $[ 0, 1, 2, 3, 5, 8 ]$ & $[0,1,4]$& 13 & 8 & 0.141 & 18, 10\\
 		\hline

 		 		3,7 	& $[ 0, 1, 2, 4, 7, 15, 16 ]$ & $[0,1,8]$& 19 & 10 & 5.312 & 21, 10\\
 		\hline

 		 3,8	& $[ 0, 1, 2, 3, 5, 7, 12, 13 ]$ & $[0,1,4]$& 19 & 8 & 0.766 & 25, 8\\
 		\hline

 		 3,9	& $[ 0, 1, 2, 3, 5, 7, 12, 13, 16 ]$ & $[0,1,4]$& 19 & 8 & 0.75 & 30, 8\\
 		\hline
 	
        4,5  & $[0, 1, 2, 4, 7]$ & $[0,1,3,9]$ & 13 & 14 & 0.5 & 23, 24 \\
		\hline

		4,6 	& $[ 0, 1, 2, 3, 6, 11 ]$ &$[0,1,4,5]$ & 17 & 18 & 3.437& 24, 8\\
		\hline

		4,7 	& $[ 0, 1, 2, 3, 5, 7, 14 ]$ &$[0,1,4,5]$ & 19 & 16 &7.625& 30, 12\\
		\hline

		4,8 	& $[ 0, 1, 2, 3, 5, 7, 12, 13 ]$ &$[0,1,4,5]$ & 19 & 10 &7.375& 39, 12\\
		\hline

		4,9 	& $[ 0, 1, 2, 3, 5, 7, 12, 13, 16 ]$ &$[0,1,4,5]$ & 19 & 10 &7.156& 48, 12\\
		\hline
 	\end{tabular}
 	\endgroup
 	\label{TabC3}
 	%\vspace{-.7em}
 \end{table} 
 \vspace{-0.15cm}
 
 \begin{table}[t]
 	\setlength{\tabcolsep}{.7 pt}
 	\centering
 	\caption{    An exponent matrix $B$  of a $(\gamma,n)$-regular compact QC-LDPC code free of 8-cycles-wc, with $g=6$, $n\geq10$ and minimum  $N$ which is compared with the minimum $N$ of a QC-LDPC code with $g=8$.}

%\vspace{-0.25cm}

\begingroup\fontsize{8.5pt}{11pt}
\begin{tabular}{|l|l|l|l|l|}
		 \hline
		 $\gamma,n$  &  Second row of $B$ & $\vec{B_1}$ &  $N$ &$N$, $g=8$ \\
		 
		 \hline
		 3,10	& $[0, 1, 2, 3, 4, 6, 8, 11, 12, 16]$& $[0,1,5]$& 21  &35\\
 		\hline

 		3,11	& $[ 0, 1, 2, 3, 4, 5, 7, 9, 12, 13, 17 ]$& $[0,1,6]$& 31  & 41\\
 		\hline
		 
		 3,12 & $[0, 1, 2, 4, 9, 10, 12, 15, 19, 21, 31, 32]$& $[0,1,3]$ & 35  & 45\\
		 \hline
		 
		 3,13 & $\left[\begin{array}{l}
		 0, 1, 2, 4, 7, 8, 9, 11, 14, 22, 25, 31, 40
		  \end{array}\right]$& $[0,1,3]$ & 43  & 50\\
		 \hline
		 
		 3,14 & $\left[\begin{array}{l}
		 0, 1, 2, 4, 7, 8, 9, 11, 14, 15, 26, 31, 35, 36
		  \end{array}\right]$& $[0,1,3]$ & 47  & 57\\
		 \hline
		 
		 3,15 & $\left[\begin{array}{l}
		 0, 1, 2, 4, 7, 8, 9, 11, 14, 16, 18, 22,35, 39, 41
		  \end{array}\right]$& $[0,1,3]$ & 53  & 63\\
		 \hline
		 
		 3,16 & $\left[\begin{array}{l}
		 0, 1, 2, 4, 7, 8, 9, 11, 14, 15, 16, 18,39, 41, 43, 47
		  \end{array}\right]$& $[0,1,3]$ & 59  & 71\\
		 \hline
		 
		 3,17 & $\left[\begin{array}{l}
		 0, 1, 2, 4, 7, 8, 9, 11, 14, 15, 16, 18,35, 38, 41, 43, 47
		  \end{array}\right]$& $[0,1,3]$ & 61  & 79\\
		 \hline
		 
		 3,18 & $\left[\begin{array}{l}
		 0, 1, 2, 4, 7, 8, 9, 11, 15, 18, 20, 23, 32, 41, 47, 49, 51, 62
		  \end{array}\right]$& $[0,1,3]$ & 67  & 88\\
		 \hline
		 
		  4,10 & $[0, 1, 2, 4, 7, 8, 9, 13, 17, 23]$& $[0,1,3,4]$ & 37  & 57\\
		 \hline
		 
 4,11 & $[0, 1, 2, 4, 7, 8, 9, 13, 22, 25, 28]$& $[0,1,3,4]$ & 47  & 67\\
		 \hline
		 
		 4,12 & $[0, 1, 2, 15, 16, 17, 21, 24, 25, 33, 36, 48]$& $[0,1,3,50]$ & 52  & 80\\
		 \hline
 	\end{tabular}
 	\endgroup
 	\label{TabC4}
 	%\vspace{-.7em}
 \end{table}

\subsection{  Submatrices of $B$ of multiple-edge protographs}\label{B} 
  According to \cite{Sidon}, to  consider  6-cycles in the exponent matrix of an ME-Ptg, Equation (\ref{Equation})  has to be checked for all submatrices of sizes $1\times1,\ 1\times 2,\ 2\times1, 1\times3,\ 3\times1,\ 2\times2,\ 2\times3,\ 3\times2$ and $3\times3$. In this case, if the TG  is free of 8-cycles-wc, then there are no pairs like $\left(b^{r^{}_i}_{m_{i}n_{i}}, b^{r^{\prime}_i}_{m_{i}n_{i+1}}\right)$ which occur in two 6-cycle equations. 
 \begin{Example}
  The base matrix of a Ptg-based Raptor-like (PBRL) code with an ME-Ptg consists of two submatrices $W_{HRC}$ and $W_{IRC}$ which yield  $\small{W=\left[\begin{array}{cc}
				W_{HRC} & 0\\
				W_{IRC} & I
				\end{array}\right]}$, where $0$ and $I$ refer to all-zeros and identity matrix of appropriate size. Applying the above mentioned restriction to a PBRL-LDPC code in \cite{Divsalar} with $\tiny{W_{HRC}=\left[\begin{array}{cccccc}
				1& 1& 2& 1& 2& 1\\
				2& 2& 1& 2& 1& 2\\
				\end{array}\right]}$ and $\tiny{W_{IRC}=\left[\begin{array}{cccccc}
				1& 1& 1& 1& 1& 1\\
				1& 1& 1& 0& 1& 0\\
				0& 1& 0& 0& 1& 1\\
				1& 0& 0& 1& 0& 1\\
				0& 0& 1& 0& 1& 0\\
				0& 1& 0& 1& 0& 1\\
				1& 0& 1& 0& 1& 0\\
				\end{array}\right]}$
			results in an exponent matrix with $N=78$ and the following $\small{\left[\begin{array}{c}
				B_{HRC}\\
				B_{IRC}\\
				\end{array}\right]}$. The TG of this PBRL-LDPC code, $C_6$, has girth 6 and is free of  8-cycles-wc whose performance curve is shown in Fig. \ref{Fig2}: 			
$$
  \small{\left[\begin{array}{cccccc}
(0)&(0)&(0,27) &(0)&(0,37)&(0)\\
(0,4)&(21,61)&(3) &(29,52)&(76)&(18,46)\\
(3)&(44)&(11) &(31)&(73)&(50)\\
(36)&(60)&(54) &(\infty)&(18)&(\infty)\\
(\infty) &(72)&(\infty) &(\infty) &(6)&(18)\\
63& (\infty)  & (\infty) &1 & (\infty)  & 62\\
(\infty)  & (\infty)  &33 & (\infty) &63 & (\infty) \\
(\infty)  &36 & (\infty)  &75 & (\infty)  &48 \\
36 & (\infty)  &28 & (\infty)  &44&(\infty)  \\
\end{array}\right].}
$$
\end{Example}
\begin{Example}
\textcolor{blue}A sequence $S=\{s_0,s_1,\ldots,s_l\}$ over $\mathbb{Z}_m$ with $l+1\choose2$ distinct sums $(s_i+s_j) \mod m\ (i\leq j)$ is a Sidon sequence \cite{Sidon}. For example, $S_1=\{0,6,9,10,21,23,28\}$ over $\mathbb{Z}_{48}$ is a Sidon sequence. We consider an exponent matrix whose integer entries are chosen from  $S$. For instance, $E=[E^{(0)},E^{(1)},E^{(2)}]$, where $\small{E^{(k)}=\left[\begin{array}{ccc}
(a^{(k)}_0,b^{(k)}_0) & (\infty) & (0)\\
(0) & (a^{(k)}_1,b^{(k)}_1) & (\infty)\\
(\infty) & (0) & (a^{(k)}_2,b^{(k)}_2) \\
\end{array} \right]}$ and all $a^{(k)}_i$s and $b^{(k)}_i$s are in $S$,  results in a $(3,9)$-regular ME-QC-LDPC code from Sidon sequences with $N=m$ \cite{Sidon}.  From $S_1$ over $\mathbb{Z}_{48}$ and avoiding 4-cycles and 8-cycles-wc, the pairs in the diagonal of this $E^{(k)}$ are $(0,6),(0,9),(0,10)$ for $k=0$, $(9,10),(6,23),(6,21)$ for $k=1$ and $(23,28),(21,28),(9,23)$ for $k=2$.
\end{Example} 
%\vspace{-0.25cm}
\section{Conclusion}\label{V}
We define a cycle with a chord in the TG of an LDPC code and propose the impact of removing short cycles with a chord on the lower bound of the size of small TSs and $d_{\min}$ of a code. We prove that without increasing the girth we can improve the important factors of an LDPC code which significantly influence the performance curve of a code.   For any protograph, single-edge or multiple-edge, regular or irregular, we provide  conditions to remove these cycles from the TG of a girth-6  QC-LDPC code.
%\vspace{-0.5cm}

%%%%%%%%%%%%%%%%%%%%%%%%%%%%%%%%%%%%%%%%%%%%%%

\begin{center}
	\begin{figure}
		\centering

		\includegraphics[scale=.7]{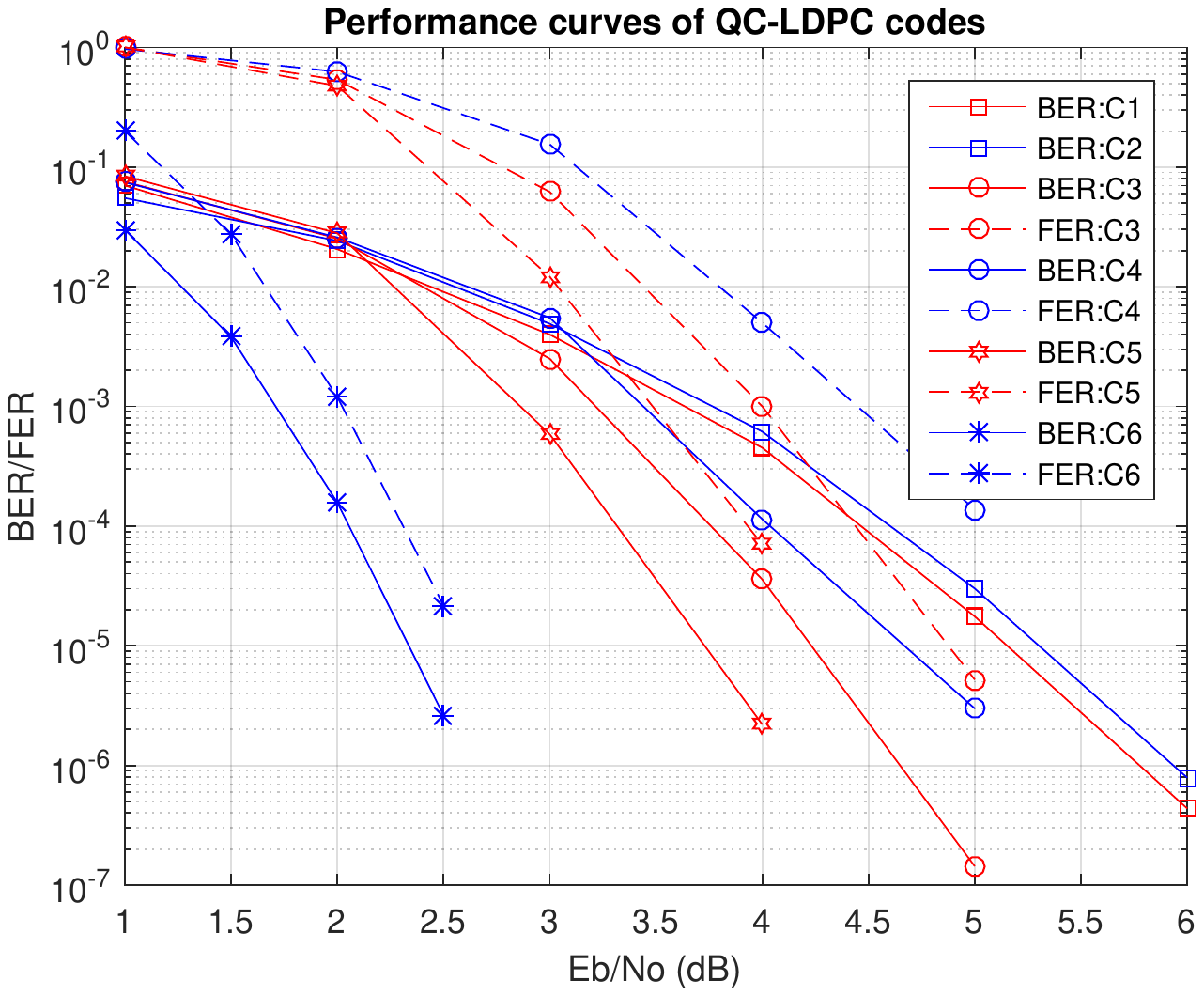}

		\caption{Bit and/or frame error rates of QC-LDPC codes, $C_1,\cdots,C_6$.} \label{Fig2}
	\end{figure}
\end{center} 
\end{document}